\documentclass[letterpaper, 10 pt, conference]{ieeeconf}  

\IEEEoverridecommandlockouts                              

\usepackage{graphics} 
\usepackage{epsfig} 
\usepackage{mathptmx} 
\usepackage{times} 
\usepackage{amsmath} 
\usepackage{amssymb}  
\usepackage{booktabs}
\usepackage[caption=false,font=footnotesize]{subfig}

\newcommand{\tp}{\normalfont \text{T}}

\newcommand{\Val}{\operatorname{Val}}

\newtheorem{remark}{Remark}

\newtheorem{cor}{Corollary}

\newtheorem{theorem}{Theorem}
\newtheorem{assumption}{Assumption}

\title{\LARGE \bf \texttt{Flip-Team}:  Cooperative Takeover Games with Stochastic Human Override
}

\author{Sandeep Banik$^{1}$ and Naira Hovakimyan$^{1}$
\thanks{$^{1}$Both the authors are affiliated with the department of Mechanical Science and Engineering, University of Illinois Urbana-Champaign, Illinois, USA; email:
        {\tt\small baniksan@illinois.edu, nhovakim@illinois.edu}}}

\begin{document}

\maketitle
\thispagestyle{empty}
\pagestyle{empty}

\begin{abstract}
Shared autonomy requires principled mechanisms for allocating and transferring control between a human and an autonomous agent.
Existing approaches often rely on blending control inputs or heuristic switching rules, which lack theoretical guarantees and fail to account for the dynamics of authority transfer.
This paper develops a cooperative game-theoretic framework for authority switching in shared autonomy.
We formulate the control switching problem as an identical-interest dynamic game in which authority transitions are embedded into the system dynamics, yielding optimal switching policies rather than ad hoc rules.
We establish the existence and characterization of team-optimal policies in pure strategies under stochastic human override, accounting for asymmetric authority where humans retain override capability.
For linear-quadratic systems, we derive closed-form recursions for the optimal switching policies and value functions, enabling efficient computation independent of the continuous state.
We validate the framework on scalar and multi-dimensional linear systems, demonstrating how optimal switching adapts to varying system dynamics, cost structures, and override probabilities.
The results reveal fundamental trade-offs between human adaptability and autonomous efficiency, illustrating the practical benefits of grounding shared autonomy in cooperative game theory.
\end{abstract}

\section{Introduction}\label{sec:Intro}

Modern cyber–physical systems (CPS) increasingly rely on \textit{shared autonomy}, where humans and automated agents dynamically share control authority. 
Such shared paradigm is critical in safety-critical domains, such as autonomous driving, assistive robotics, power grid operations, where autonomous systems excel at computation but lack adaptability, while humans provide contextual reasoning but may be overwhelmed or out of the loop~\cite{javdaniSharedAutonomyHindsight2015, nikolaidisHumanRobotMutualAdaptation2017}. 
Failures to coordinate takeovers have been implicated in real-world incidents from autopilot disengagement in aviation~\cite{gouraudAutopilotMindWandering2017} to autonomous vehicle crashes~\cite{alambeigiCrashThemesAutomated2020}, highlighting the need for principled frameworks for cooperative takeover.
This work adopts a cooperative formulation in which both human and autonomy act as a unified team pursuing a common mission objective, while retaining distinct roles: the human provides adaptive oversight with override capability, and the autonomous system offers consistent, low-cost control.

Early approaches to shared autonomy typically combined human and robot commands through a human-in-the-loop paradigm using control blending or arbitration mechanisms~\cite{draganPolicyblendingFormalismShared2013}. 
In these methods, the robot predicts the human’s intended goal and then determines the level of assistance to provide~\cite{aarnoMotionIntentionRecognition2008}. 
A common strategy treats human input as a noisy estimate of intent and computes a weighted blend with the robot's action~\cite{loseyReviewIntentDetection2018,abbinkTopologySharedControl2018}, implemented via Bayesian inference~\cite{aarnoMotionIntentionRecognition2008}, POMDP planning~\cite{javdaniSharedAutonomyHindsight2015}, or linear arbitration.
Critically, these blending approaches require continuous human engagement, namely, a human-in-the-loop paradigm, where the operator must provide input at every time step.
This places substantial cognitive burden on the human, limiting scalability to prolonged or multi-task operations.
In contrast, the present work adopts a human-on-the-loop paradigm, where the autonomous system operates independently while the human retains supervisory override authority, intervening only when necessary.
While blending methods allow users to complete tasks more efficiently, they rely heavily on accurate intent prediction and hand-tuned blending rules, often retaining task-specific heuristics without guarantees of stability or generalization~\cite{jeonSharedAutonomyLearned2020, reddySharedAutonomyDeep2018,javdaniSharedAutonomyHindsight2015}.

Beyond blending, two alternative paradigms have emerged. Game-theoretic approaches model shared autonomy as an interaction between decision-making agents, capturing mutual influence and strategic adaptation~\cite{liDifferentialGameTheory2019, sadighPlanningAutonomousCars2016, nikolaidisGameTheoreticModelingHuman2017}. However, most formulations assume symmetric roles or require full knowledge of human utility, which is often unrealistic. A separate line of work focuses on authority switching, where explicit decisions determine who has control at a given time. Frameworks for adjustable autonomy and mixed-initiative control dynamically allocate authority based on task demands, communication delays, or operator workload~\cite{scerriAdjustableAutonomyReal2002}, while adaptive autonomy shifts control according to human performance, trust, or safety thresholds~\cite{loseyPhysicalInteractionCommunication2022}. These approaches emphasize flexibility but remain heuristic, lacking principled criteria for when to switch. The present work addresses this gap by providing a cooperative framework that yields optimal switching policies under asymmetric authority, where the human retains override capability.

A related body of work on resource takeover games provides the mathematical foundation for this study. The \texttt{FlipDyn} framework introduced the idea of modeling authority over a dynamical system as a competitive game in which two players strategically “flip” control to minimize or maximize long-term performance~\cite{banikFlipDynGameResource2022a}.
The \texttt{FlipDyn} model, inspired by the \texttt{FlipIt} game~\cite{vandijkFlipItGameStealthy2013} of stealthy takeovers, formalized how equilibrium takeover strategies emerge when agents contest authority over linear–quadratic systems.
Subsequent extensions generalized the framework to networked dynamical systems, capturing how takeovers propagate across interconnected nodes and deriving equilibrium conditions for graph-structured interactions~\cite{zhu2015game}.
Other studies explored the role of cyber–physical resilience, applying game-theoretic reasoning to capture adversarial control switching in large-scale infrastructures~\cite{banikFlipDynGraphsResource2025}.
Collectively, these works establish rigorous methods for analyzing control takeovers in adversarial contexts. 
The present paper departs from the adversarial setting by translating these mathematical tools to \textit{cooperative shared autonomy}, where both agents pursue aligned objectives but still require principled strategies for when to yield or assert control.
We term this extension \texttt{Flip-Team}.
The framework applies to systems with explicit authority handoffs where control is mutually exclusive at each instant -- such as teleoperation, supervisory control, and semi-autonomous navigation. The main contributions are:
\begin{enumerate}
    \item \textbf{Cooperative takeover formulation}:
    We formulate the human–autonomy takeover problem as an identical-interest dynamic game in which the switching decision is embedded into the system dynamics.
    This formulation captures asymmetric authority, stochastic human override, and state-dependent costs within a unified framework.
    
    \item \textbf{Characterization of optimal switching}: 
    We establish the existence of team-optimal switching policies and characterize them in the space of pure strategies under stochastic human override.
    The analysis yields explicit threshold conditions for when authority transfer occurs in general dynamical systems.
    
    \item \textbf{Analytical solutions for linear–quadratic systems}: 
    For the class of linear–quadratic (LQ) systems, we derive closed-form recursions for the optimal switching policies and value functions.
    These results enable efficient computation of cooperative takeover policies in continuous state spaces, with switching thresholds that are \emph{independent of the continuous state}.
\end{enumerate}

We illustrate our results for the scalar and $n-$dimensional system. 
Section~\ref{sec:Problem_formulation} formalizes the cooperative takeover problem as an identical-interest dynamic game with switching dynamics. Section~\ref{sec:FlipCoop_General_Systems} establishes the existence and characterization of optimal switching policies, followed by closed-form analytical results for the linear–quadratic setting in Section~\ref{sec:FlipCoop_Linear_Systems}. 
Section~\ref{sec:Evaluation} demonstrates the framework on a scalar and 2D vehicle lane tracking task and Section~\ref{sec:Conclusion} summarizes the findings and discusses directions for future work.

\section{Problem Formulation}\label{sec:Problem_formulation}
Consider a discrete-time dynamical system controlled either by a human or an autonomous agent. 
The \texttt{FlipDyn} state, $\alpha_{k} \in \{\text{H},\text{A}\}$ indicates whether the human ($\alpha_k = \text{H}$) or the autonomous agent ($\alpha_k = \text{A}$) has control over the system at time $k$. 
The state evolution of the system controlled by the human is given by:
\begin{align}\label{eq:human_dynamics}
	x_{k+1} = F_{k}^{\text{H}}(x_k,u_k),
\end{align}
where $k$ denotes the discrete-time index, taking values from the integer set $\mathcal{K} := \{1,2,\dots, L\} \subset \mathbb{N}$, $x_{k} \in \mathbb{R}^{n}$ is the state of the system, $u_{k} \in \mathbb{R}^m$ is the control input of the system, and $F_{k}^{\text{H}}: \mathbb{R}^{n} \times \mathbb{R}^{m} \rightarrow \mathbb{R}^{n}$ is the state transition function.
Similarly, the state evolution under an autonomous agent is given by:
\begin{align}\label{eq:autonomous_dynamics}
	x_{k+1} = F_{k}^{\text{A}}(x_k,w_k),
\end{align}
where $F_{k}^{\text{A}}: \mathbb{R}^{n} \times \mathbb{R}^{p} \rightarrow \mathbb{R}^{n}$ is the state transition function under the autonomous agent with control input $w_k \in \mathbb{R}^{p}$. 
We describe a takeover through the action $\pi^{j}_k \in \{0,1\}$, which denotes the action of the agent $j \in \{\text{H},\text{A}\}$ at time $k$, where $j = \text{H}$ denotes the human and $j=\text{A}$ denotes the autonomous agent.
The action $\pi_{k}^{j} = 1, \forall j$ corresponds to takeover/request to takeover and $\pi_{k}^{j} = 0$ otherwise (idle). The binary \texttt{FlipDyn} state updates based on the agent's action and prior \texttt{FlipDyn} state. Given $\alpha_{k} = \text{H}$, the \texttt{FlipDyn} state at time $k+1$ is:
\begin{align}\label{eq:flip_state_H}
	\alpha_{k+1} &= \begin{cases}
		\alpha_{k}, & \text{if } \{\pi^{\text{H}}_{k} = 0,  \pi^{\text{A}}_{k} = 0\}, \\
        \text{H}, & \text{if } \{\pi^{\text{H}}_{k} = 0,  \pi^{\text{A}}_{k} = 1\}, \\
        \text{H}, & \text{if } \{\pi^{\text{H}}_{k} = 1,  \pi^{\text{A}}_{k} = 0\}, \\
        \text{A}, & \text{otherwise},
	\end{cases}
\end{align}
where $\pi^{\text{H}}_k = 0$ corresponds to human agent being idle or retaining the control of the system and $\pi^{\text{H}}_k = 1$ represents request to takeover.
The first three cases correspond to human retaining control; the last corresponds to consensual handoff to autonomy.
Similarly, given $\alpha_{k} = \text{A}$,  \texttt{FlipDyn} state update is:
\begin{align}\label{eq:flip_state_A}
	\alpha_{k+1} &= \begin{cases}
		\alpha_{k}, & \text{if } \{\pi^{\text{H}}_{k} = 0,  \pi^{\text{A}}_{k} = 0\}, \\
        \text{A}, & \text{if } \{\pi^{\text{H}}_{k} = 0,  \pi^{\text{A}}_{k} = 1\}, \\
        \text{H}, & \text{if } \{\pi^{\text{H}}_{k} = 1,  \pi^{\text{A}}_{k} = 0\} \text{ with probability } p_{k}, \\
        \text{A}, & \text{if } \{\pi^{\text{H}}_{k} = 1,  \pi^{\text{A}}_{k} = 0\} \text{ with probability } 1-p_{k}, \\
        \text{H}, & \text{otherwise},
	\end{cases}
\end{align}
where $\pi^{\text{H}}_k = 1$ corresponds to human agent takeover and $\pi^{\text{A}}_k = 1$ corresponds to request to takeover.
The parameter $p_k \in (0,1)$ represents the probability that a human override attempt succeeds, given the human has chosen to intervene ($\pi^{\text{H}}_k = 1$) and the autonomous agent has not requested handoff ($\pi^{\text{A}}_k = 0$).
This captures authority asymmetry arising from interface latency, arbitration mechanisms, or system design.
The autonomous agent retains control of the system deterministically when human agent is idle (second condition) and with probability $1-p_{k}$ (fourth condition) when human agent chooses to takeover.
The human agent takes over with probability $p_k$ when the autonomous agent does not request to takeover ($\pi^{\text{A}}_k = 0$) and deterministically takes over when both $\pi^{\text{A}}_k = \pi^{\text{H}}_k = 1$. Notice that the difference between~\eqref{eq:flip_state_H} and~\eqref{eq:flip_state_A} is the human agent being uncertain with probability $p_k$ and ability to takeover the system irrespective of the autonomous agent's actions.
Such a model captures the current design of shared control system where the human agent has higher authority over autonomous agents. 

Takeovers are mutually exclusive, i.e., only one agent is in control of the system at any given time. The continuous state $x_{k+1}$ at time $k+1$ is dependent on $\alpha_{k+1}$. 
In this work, we aim to solve for a takeover strategy for both the human and autonomous agent. Given a non-zero initial state $x_{1}$, we pose the takeover problem as an identical interest dynamic game described by the dynamics~\eqref{eq:human_dynamics},~\eqref{eq:autonomous_dynamics},~\eqref{eq:flip_state_H} and~\eqref{eq:flip_state_A} over a finite-horizon $L$, where both the human and autonomous agent aim to minimize a net cost given by:
\begin{equation}\label{eq:obj_def_alpha}
	\begin{aligned}
		J(x_{1}, \alpha_{1}, \{\pi^{\text{H}}_{\mathbf{L}}\}, \{\pi^{\text{A}}_{\mathbf{L}}\}) & = \mathbb{E}_{\mathbf{p}}\left[g_{L+1}(x_{L+1}, \alpha_{L+1}) + \sum_{t=1}^{L} g_t(x_t, \alpha_t) \right. \\ & \left. + \pi_{t}^{\text{H}|\alpha_{t}}h_t(x_t) + \pi^{\text{A}|\alpha_{t}}_{t}a_t(x_t) \right], 
	\end{aligned}
\end{equation}
where the expectation is over the stochastic authority transitions induced by the override probability sequence $\mathbf{p} := \{p_1, \dots, p_{L}\}$. The takeover sequence by both the human and autonomous agent is represented as $\{\pi_{\mathbf{L}}^j\} := \{\pi_1^{j|\alpha_{1}}, \dots, \pi_{L}^{j|\alpha_{L}}\}$ for $j \in \{\text{H},\text{A}\}$.
The state cost $g_t(x_t, \alpha_t): \mathbb{R}^{n} \times \{\text{H},\text{A}\} \rightarrow \mathbb{R}$ captures control performance under each agent (e.g., tracking error, energy), while $h_t(x_t), a_t(x_t): \mathbb{R}^{n} \rightarrow \mathbb{R}$ are takeover costs representing cognitive load, attention switching, or transition risk for human and autonomy respectively. Asymmetric costs $g^{\text{H}}_t \neq g^{\text{A}}_t$ and $h_t \neq a_t$ encode differences in control effectiveness and takeover burden without requiring separate utility functions.
We term the dynamic game~\eqref{eq:obj_def_alpha} between the human and autonomous agent as \emph{\texttt{Flip-Team}}, where both agents jointly optimize the takeover strategy.
In particular, we consider a subclass of games known as identical interest games~\cite{hespanha2017noncooperative}, where both agents are aligned in minimizing a common cost function~\eqref{eq:obj_def_alpha}. 

Over finite-horizon $L$, the optimal joint policy $(\pi_{\mathbf{L}}^{\text{H}*}, \pi_{\mathbf{L}}^{\text{A}*})$ minimizes the total cost:
\begin{equation*}
    (\pi_{\mathbf{L}}^{\text{H}*}, \pi_{\mathbf{L}}^{\text{A}*}) = \arg\min_{\pi_{\mathbf{L}}^{\text{H}}, \pi_{\mathbf{L}}^{\text{A}}} J(x_1, \alpha_1, \pi_{\mathbf{L}}^{\text{H}}, \pi_{\mathbf{L}}^{\text{A}}).
\end{equation*}
This team-optimal solution is also a Nash equilibrium: neither agent can reduce cost by unilaterally deviating from the joint optimum~\cite{hespanha2017noncooperative}. 

\begin{remark}[Relation to MDPs]\label{rem:mdp}
Since both agents minimize a common cost, the \texttt{Flip-Team} game is strategically equivalent to a multi-agent MDP.
We retain the game-theoretic formulation because (i) it yields closed-form switching thresholds (Section~\ref{sec:FlipCoop_General_Systems}, Theorem~\ref{th:NE_Val_gen_FDC}) that provide interpretable design criteria, and (ii) it naturally extends to settings with information asymmetry.
\end{remark}
In the next section, we derive the optimal switching policies and the conditions under which authority transfers occur.

\section{\texttt{Flip-Team} for general systems}\label{sec:FlipCoop_General_Systems}
We will begin by deriving the optimal switching policies of the \texttt{Flip-Team} game for any given control policy pair  $u_{\mathbf{L}}, w_{\mathbf{L}}$.
Our approach begins by defining the value function.
\subsection{Value Function}
Given a \texttt{FlipDyn} state at time $k \in \mathcal{K}$, the value comprises an instantaneous state cost and an additive cost-to-go based on the players takeover actions.
The cost-to-go is determined via a cost-to-go matrix in each \texttt{FlipDyn} state $\alpha_{k} = \text{H}$ and $\alpha_{k} = \text{A}$, represented by $\Xi_{k+1}^{\text{H}} \in \mathbb{R}^{2 \times 2}$ and $\Xi_{k+1}^{\text{A}} \in \mathbb{R}^{2 \times 2}$, respectively.
Let $V^{\text{H}}_k(x, \Xi_{k+1}^{\text{H}})$ and $V^{\text{A}}_k(x, \Xi_{k+1}^{\text{A}})$ be the value functions of the continuous state $x$ and cost-to-go matrices for $\alpha_k = \text{H}$ and $\alpha_k = \text{A}$, respectively. The entries of the cost-to-go matrix $\Xi^{\text{H}}_{k+1}$ corresponding to each pair of takeover actions are given by:
\begin{equation}\label{eq:Cost_to_go_H}
    \begin{aligned}
		& \begin{matrix} & \hphantom{000} \text{Idle} & & \hphantom{v_{k+1}^0(.,.} \text{Takeover}\end{matrix} \\
		\begin{matrix} \text{Idle} \\[5pt] \text{Request to}\\\text{takeover} \end{matrix} & \underbrace{\begin{bmatrix}
			v_{k+1}^{\text{H}} &  v_{k+1}^{\text{H}} + a_k(x)  \\[8pt]
			v_{k+1}^{\text{H}} + h_k(x) &  v_{k+1}^{\text{A}} + h_k(x) + a_k(x) \\[5pt] 
		\end{bmatrix}}_{\Xi_{k+1}^{\text{H}}}
	\end{aligned},
\end{equation}
\begin{align}
    \label{eq:V_k_0} \text{where } \ & v_{k+1}^{\text{H}} := V_{k+1}^{\text{H}}\left(F_k^{\text{H}}(x,u_k),\Xi_{k+2}^{\text{H}}\right), \\
    \label{eq:V_k_1} & v_{k+1}^{\text{A}} := V_{k+1}^{\text{A}}(F_k^{\text{A}}(x,w_k),\Xi_{k+2}^{\text{A}}).
\end{align}
The row and column entries of $\Xi^{\text{H}}_{k+1}$ are based on the player's actions, described by~\eqref{eq:flip_state_H},~\eqref{eq:flip_state_A} and its associated dynamics~\eqref{eq:human_dynamics},~\eqref{eq:autonomous_dynamics}.
We will refer $X(i,j)$ as the $(i,j)$-th entry of the matrix $X$. The first row entries $\Xi_{k+1}^{\text{H}}(1,1)$ and $\Xi_{k+1}^{\text{H}}(1,2)$ correspond to the human agent remaining idle, which prevents the autonomous agent to takeover despite the column action of takeover. The second row entries $\Xi_{k+1}^{\text{H}}(2,1)$ and $\Xi_{k+1}^{\text{H}}(2,2)$ correspond to action of request to takeover, and transition to the autonomous agent ($v^{\text{A}}_{k+1}$) only when the autonomous agent is ready (column action of takeover).
The entries of $\Xi_{k+1}^{\text{H}}$ couple the  value in each \texttt{FlipDyn} state. 
At time $k$ for a given human agent control policy $u_k$, state $x$ and $\alpha_k=\text{H}$, the value function satisfies:
\begin{equation}
    \label{eq:V_k^0_cost_to_go}
	V^{\text{H}}_k(x, \Xi_{k+1}^{\text{H}}) = g_k(x,\text{H})  + \Val(\Xi^{\text{H}}_{k+1}), 
\end{equation}
where $\Val(X_{k+1}^{\alpha_{k}}):= \min_{y_{k}^{\alpha_{k}}} \min_{z_{k}^{\alpha_{k}}} y_{k}^{{\alpha_{k}}^{\tp}}X_{k+1}z_{k}^{\alpha_{k}}$ represents the  optimal value of the cost-to-go matrix $X_{k+1}$ for the \texttt{FlipDyn} state $\alpha_{k}$, and $\Xi^{0}_{k+1} \in \mathbb{R}^{2 \times 2}$ is the cost-to-go matrix.
Since both agents minimize a common cost, the optimal action pair selects the entry of $\Xi^{\text{H}}_{k+1}$ with minimum cost-to-go from state $x$ at time $k$.

Similarly, for $\alpha_k = \text{A}, \forall k$, the cost-to-go matrix entries $\Xi_{k+1}^{\text{A}}$ are:
\begin{equation}\label{eq:Cost_to_go_A}
    \begin{aligned}
		& \begin{matrix} & \hphantom{0000} \text{Idle} & & \hphantom{v_{k+1}^0000} \text{Request to takeover}\end{matrix} \\
		\begin{matrix} \text{Idle} \\[10pt] \text{Takeover}\end{matrix} & \underbrace{\begin{bmatrix}
			v_{k+1}^{\text{A}} \hphantom{000}  &  \hphantom{00} v_{k+1}^{\text{A}} + a_k(x) \hphantom{0}  \\[5pt]
			\begin{matrix}
                    p_{k}v_{k+1}^{\text{H}} + (1-p_{k})v_{k+1}^{\text{A}} \\
                        + h_k(x)       
                \end{matrix}
              &  \begin{matrix}
                    v_{k+1}^{\text{H}} + h_k(x) \\
                        + a_k(x)       
                \end{matrix}
		\end{bmatrix}}_{\Xi_{k+1}^{\text{A}}}.
	\end{aligned}
\end{equation}
The first row corresponds to autonomy retaining control (human idle). Entry $\Xi_{k+1}^{\text{A}}(2,1)$ reflects stochastic human override with probability $p_k$; entry $\Xi_{k+1}^{\text{A}}(2,2)$ reflects deterministic human takeover by mutual agreement.
Analogous to~\eqref{eq:V_k^0_cost_to_go} the  value for $\alpha_{k} = \text{A}$ satisfies:
\begin{equation}
    \label{eq:V_k^1_cost_to_go}
	\text{with} \ V^{\text{A}}_k(x, \Xi_{k+1}^{\text{A}}) = g_k(x,\text{A}) +  \Val(\Xi^{\text{A}}_{k+1}).
\end{equation}

With the  values established in each of the \texttt{FlipDyn} states, in the following subsection, we will characterize the optimal switching policies and value functions over the horizon $L$.

\subsection{Optimal switching policy}
In order to characterize the optimal value, we restrict the cost functions to belong to a finite domain, stated in the following assumption.
\begin{assumption}\label{ast:general_costs}
    [Non-negative costs] At any time instant $k \in \mathcal{K}$, the state and takeover costs $g_k(x,\alpha), h_k(x), a_k(x)$, for all $x \in \mathbb{R}^{n},$ and $\alpha \in \{\text{H},\text{A}\}$ are non-negative $(\mathbb{R}_{\geq 0})$. 
\end{assumption}

Assumption~\ref{ast:general_costs} ensures that cost comparisons within the cost-to-go matrix are sign-consistent, allowing for a clear characterization of optimal policies.
Under Assumption~\ref{ast:general_costs}, we derive the  value and takeover policies for the finite time-horizon in the following result.

\begin{theorem}\label{th:NE_Val_gen_FDC}
    (\textbf{Case $\alpha_k = \text{H}$}) Under Assumption~\ref{ast:general_costs} and known control policies, $u_{\mathbf{L}}$ and $w_{\mathbf{L}}$, the optimal switching policies of the \texttt{Flip-Team} game~\eqref{eq:obj_def_alpha} at time $k \in \mathcal{K}$, subject to the continuous state dynamics~\eqref{eq:human_dynamics},~\eqref{eq:autonomous_dynamics} and \texttt{FlipDyn} dynamics~\eqref{eq:flip_state_H},~\eqref{eq:flip_state_A}, are given by:
    \begin{align}
    \begin{split}\label{eq:TP_gen_H}
            \{\pi^{\text{H}*|\text{H}}_{k}, \pi^{\text{A}*|\text{H}}_{k}\} = \begin{cases}
            \{0,0\} (\text{retain}), & \text{if }  
                \begin{matrix}
                    \tilde{v}_{k+1} + h_{k}(x) \\ + a_{k}(x)
                \end{matrix} > 0,
             \\[5 pt]
            \{1,1\} (\text{handoff}), & \text{otherwise}.
            \end{cases} 
    \end{split}
    \end{align}
    
    The  value is given by:
    \begin{align}\label{eq:Val_gen_H}
        v_{k}^{\text{H}} = 
        \begin{cases}
            \begin{aligned}
				& g_k(x,\text{H}) + v_{k+1}^{\text{H}},
            \end{aligned} &\text{if } \tilde{v}_{k+1} + h_{k}(x) + a_{k}(x) > 0, \\
            \begin{aligned}
				& g_k(x,\text{H}) + v_{k+1}^{\text{A}} \\ & + h_{k}(x) + a_{k}(x),
            \end{aligned} &\text{otherwise},
		\end{cases} 
	\end{align}
    where $\tilde{v}_{k+1} := v_{k+1}^{\text{A}} - v_{k+1}^{\text{H}}$.
    
    (\textbf{Case $\alpha_k = \text{A}$}) The optimal switching policies are 
    \begin{align}
    \begin{split}\label{eq:TP_gen_A}
            \{\pi_{k}^{\text{H}*|\text{A}},\pi_{k}^{\text{A}*|\text{A}}\}  = \begin{cases}
            \{0,0\} (\text{retain}), & \text{if } \ 
                \begin{matrix}
                    \tilde{v}_{k+1} < \dfrac{h_{k}(x)}{p_{k}}
                \end{matrix}, \\[10pt]
            \{1,0\} (\text{overrride}), & \text{if } \ 
                \begin{matrix}
                    \tilde{v}_{k+1} \leq \dfrac{a_{k}(x)}{1 - p_{k}}  \\
                    \tilde{v}_{k+1} \geq \dfrac{h_{k}(x)}{p_{k}},
                \end{matrix} \\
            \{1,1\} (\text{handoff}), & \text{otherwise.}
            \end{cases} 
    \end{split}
    \end{align}
    The  value is given by:
    \begin{align}\label{eq:Val_gen_A}
        v_{k}^{\text{A}} = 
        \begin{cases}
            \begin{aligned}
				& g_k(x,\text{A}) + v_{k+1}^{\text{A}},
            \end{aligned} &\text{if } \begin{matrix}
                    \tilde{v}_{k+1} < \dfrac{h_{k}(x)}{p_{k}}
                \end{matrix}, \\[10pt]
            \begin{aligned}
				& g_k(x,\text{A}) + p_{k}v^{\text{H}}_{k+1} + \\[5pt] & (1 - p_{k})v^{\text{A}}_{k+1} + h_k(x),
            \end{aligned} &\text{if } \begin{matrix}
                    \tilde{v}_{k+1} \leq \dfrac{a_{k}(x)}{1 - p_{k}}  \\
                    \tilde{v}_{k+1} \geq \dfrac{h_{k}(x)}{p_{k}},
                \end{matrix} \\
            \begin{aligned}
                g_k(x,\text{A}) + v_{k+1}^{\text{H}} + \\
                h_{k}(x) + a_{k}(x),
            \end{aligned}
             & \text{otherwise}.
		\end{cases} 
	\end{align}
    The boundary condition at $k = L$ is given by:
    \begin{gather}\label{eq:b_cond_NE_Val_gen_FDC}
        \Xi_{L+2}^{\text{H}} := \mathbf{0}_{2 \times 2}, \ \Xi_{L+2}^{\text{A}} := \mathbf{0}_{2 \times 2}, 
    \end{gather}
    where $\mathbf{0}_{i \times j} \in \mathbb{R}^{i \times j}$ represents a matrix of zeros.
\end{theorem}

\begin{proof}
    We will only derive the optimal switching policies and  value for the case of $\alpha_k = \text{A}$. We leave out the derivations for $\alpha = \text{H}$ as they are analogous to $\alpha = \text{A}$. There are three candidate optimal action pairs to consider for the $2 \times 2$ cost-to-go matrix game defined by the matrix in~\eqref{eq:Cost_to_go_A}.
    
    i) \underline{Pure strategy - \{Idle, Idle\}}:
    Both the human and autonomous agent choose the action of staying idle. 
    \noindent We determine the conditions under which such a pure policy is feasible. Under Assumption~\ref{ast:general_costs}, we compare the entries of $\Xi_{k+1}^{\text{A}}$ when the human agent opts to remain idle to obtain the condition: 
    \begin{equation}
        \begin{aligned}
            v_{k+1}^{\text{A}} < v_{k+1}^{\text{A}} + a_{k}(x),
        \end{aligned}
    \end{equation}
    which indicates that the autonomous agent always chooses to play idle. Next, we compare the entry of $ v_{k+1}^{\text{A}}$ against a human agent takeover action to obtain:
    \begin{equation*}
        \begin{aligned}
            v_{k+1}^{\text{A}} & \leq p_{k}v_{k+1}^{\text{H}} + (1 - p_{k})v_{k+1}^{\text{A}} + h_{k}(x), 
        \end{aligned}
    \end{equation*}
    \begin{equation*}
        \begin{aligned}
         \Rightarrow & v_{k+1}^{\text{A}} - v_{k+1}^{\text{H}} \leq \dfrac{h_k(x)}{p_{k}}.
        \end{aligned}
    \end{equation*}
    This policy is optimal when the condition is satisfied.
    The value, corresponding to the pure strategy of both agents remaining idle, is the entry $\Xi_{k+1}^{\text{A}}(1,1)$, given by:
    \begin{equation*}
        V_{k}^{\text{A}}(x,\Xi_{k+1}^{\text{A}}) = g_k(x,\text{A}) + v_{k+1}^{\text{A}}.
    \end{equation*}
    
    ii) \underline{Pure strategy - \{Takeover, Idle\}}:
    The autonomous agent chooses to stay idle whereas the human agent chooses to takeover. 
    \noindent To realize such a pure strategy, we compare the entry $\Xi_{k+1}^{\text{A}}(2,1)$ against $\Xi_{k+1}^{\text{A}}(1,1)$ to obtain the condition:
    \begin{equation*}
        \begin{aligned}
            v_{k+1}^{\text{A}} - v_{k+1}^{\text{H}} \geq \dfrac{h_{k}(x)}{p_{k}}.
        \end{aligned}
    \end{equation*}
    Similarly, comparing $\Xi_{k+1}^{\text{A}}(2,1)$ against $\Xi_{k+1}^{\text{A}}(2,2)$ yields:
    \begin{equation*}
        \begin{aligned}
            p_{k}v_{k+1}^{\text{H}} + (1 - p_{k})v_{k+1}^{\text{A}} + h_{k}(x) & \leq v_{k+1}^{\text{H}} + h_{k}(x) + a_{k}(x), \\
            \Rightarrow v_{k+1}^{\text{A}} - v_{k+1}^{\text{H}} & \leq \dfrac{a_{k}(x)}{1 - p_{k}}.
        \end{aligned}
    \end{equation*}
    If the conditions derived for the strategy are satisfied, then such a strategy corresponds to an optimal policy. 
    The expected value corresponds to the entry $\Xi_{k+1}^{\text{A}}(1,2)$, given by:
    \begin{equation*}
        V_{k}^{\text{A}}(x,\Xi_{k+1}^{\text{A}}) = g_k(x,\text{A}) + p_{k}v_{k+1}^{\text{H}} + (1 - p_{k})v_{k+1}^{\text{A}} + h_{k}(x).
    \end{equation*}

    iii) \underline{Pure strategy - \{Takeover, Request to takeover\}}: The human agent takes over and the autonomous agent requests to takeover. Following the same process, we compare the entry $\Xi_{k+1}^{\text{A}}(2,2)$ against $\Xi_{k+1}^{\text{A}}(2,1)$ to obtain the condition:
    \begin{equation*}
        v_{k+1}^{\text{A}} - v_{k+1}^{\text{H}} > \frac{a_{k}(x)}{1 - p_{k}}.
    \end{equation*}
    Notice, we didn't compare $\Xi_{k+1}^{\text{A}}(2,2)$ to $\Xi_{k+1}^{\text{A}}(1,2)$, because the entry $\Xi_{k+1}^{\text{A}}(1,2)$ is non-optimal.
    In other words, there is no incentive for either agent to select a policy which corresponds to the entry $\Xi_{k+1}^{\text{A}}(1,2)$, since $\Xi_{k+1}^{\text{A}}(1,1)$ is strictly better than $\Xi_{k+1}^{\text{A}}(1,2)$. 
    The  value for the derived policy is:
    \begin{equation*}
        V_{k}^{\text{A}}(x,\Xi_{k+1}^{\text{A}}) = g_{k}(x, \text{A}) + v_{k+1}^{\text{H}} + h_{k}(x) + a_{k}(x).
    \end{equation*}
    Collecting the optimal values for each derived policy, we obtain the value update equation over the  horizon of $L$ in~\eqref{eq:Val_gen_A}. The boundary conditions~\eqref{eq:b_cond_NE_Val_gen_FDC} imply that the  values at $k = L+1$ satisfy
    \begin{equation*}
        \begin{aligned}
            V_{L+1}^{\text{H}}(x,\mathbf{0}_{2 \times 2}) & = g_{L+1}^{\text{H}}(x,\text{H}), \\ V_{L+1}^{\text{A}}(x,\mathbf{0}_{2 \times 2}) & = g_{L+1}^{\text{A}}(x,\text{A}).
        \end{aligned}
    \end{equation*}
\end{proof}

\begin{remark}
The switching thresholds in~\eqref{eq:TP_gen_H} and~\eqref{eq:TP_gen_A} compare the cost-to-go difference $v^{\text{A}}_{k+1} - v^{\text{H}}_{k+1}$ against takeover penalties scaled by the override probability $p_k$. A takeover occurs only when this difference exceeds the threshold—not when the  values cross. The stochastic term $p_k$ modulates the threshold: higher $p_k$ (more reliable human override) lowers the barrier for human takeover.
\end{remark}
For a finite cardinality of the state space $\mathcal{X}$, fixed player policies $u_k$ and $w_k,  k \in \mathcal{K}$, and a finite-horizon $L$, Theorem~\ref{th:NE_Val_gen_FDC} yields the optimal value  of the \texttt{Flip-Team} game~\eqref{eq:obj_def_alpha}. However, the computational and storage complexities  scale undesirably with the cardinality of $\mathcal{X}$, especially in continuous state spaces. 
To mitigate such computational challenges, in the next section we will provide a parametric form of the  value for linear dynamics with quadratic costs.


\section{\texttt{Flip-Team} for LQ Problems}\label{sec:FlipCoop_Linear_Systems}
In this section, we restrict our attention to a linear dynamical system with quadratic costs (LQ problems). The dynamics of a linear system at time instant $k \in \mathcal{K}$, controlled by the human agent, satisfies:
\begin{equation}
    \label{eq:H_control_dynamics}
    \begin{aligned}
        x_{k+1} & = F_{k}^{\text{H}}(x_k, u_k) := E_kx_k + B_ku_k,
    \end{aligned}
\end{equation}
where $E_{k} \in \mathbb{R}^{n \times n}$ denotes the state transition matrix, while $B_{k} \in \mathbb{R}^{n \times m}$ represents the human agent control matrix. Similarly, the dynamics of the linear system when the autonomous agent takes over satisfies:
\begin{equation}
    \label{eq:A_control_dynamics}
    \begin{aligned}
        x_{k+1} & = F_{k}^{\text{A}}(x_k, w_k) := E_kx_k + C_kw_k,
    \end{aligned}
\end{equation}
where $C_{k} \in \mathbb{R}^{n \times p}$ signifies the autonomous agent control matrix.
The compact dynamics with the \texttt{FlipDyn} state  is given by:
\begin{align}\label{eq:linear_dynamics}
	x_{k+1} = E_{k}x_{k} + \mathbf{1}_{\text{H}}(\alpha_{k+1})B_{k}u_{k} + \mathbf{1}_{\text{A}}(\alpha_{k+1})C_{k}w_{k},
\end{align}
where $\mathbf{1}_{\mathcal{A}}: \alpha_{k+1} \to \{0,1\}$ is an indicator function, which maps to one if $\alpha_{k+1} = \mathcal{A}$ and zero otherwise.
The stage and takeover quadratic costs are:
\begin{gather}\label{eq:cost_quad}
        g_k(x,\alpha_k) = x^{\tp}G_k^{\alpha_k}x, \
        h_k(x) = x^{\tp}H_kx, \ a_k(x) = x^{\tp}A_kx, 
\end{gather}
where $G_k^{\alpha_k} \in \mathbb{S}^{n \times n}_{+}, H_k \in \mathbb{S}^{n \times n}_{+}, A_k \in \mathbb{S}^{n \times n}_{+}$ are positive definite matrices.

If we constrain the control policies of both agents to function of the continuous state $x$, then the value for each \texttt{FlipDyn} state can be represented purely as a function of the continuous state $x$ and \texttt{FlipDyn} state, in contrast to being contingent on both continuous state $x$ and the control input for each \texttt{FlipDyn} state.
This restriction is formally stated in the subsequent assumption.

\begin{assumption}\label{ast:linear_control_space}
    We restrict the control policies to be linear state-feedback in the continuous state $x$, described as:
    \begin{equation}\label{eq:linear_FD_control}
        u_k(x) := K_kx, \quad w_k(x) := W_kx,
    \end{equation}
    where $K_k \in \mathbb{R}^{m \times n}$ and $W_k \in \mathbb{R}^{p \times n}$ are human and autonomous agent control gains matrices, respectively. 
\end{assumption}

Under Assumption~\ref{ast:linear_control_space}, the human and autonomous agent dynamics can be compactly written as:
\begin{subequations}\label{eq:dynamics_HA_compact}
    \begin{align}
        x_{k+1} = \tilde{B}_{k}x_{k} := (E_{k} + B_{k}K_{k})x_{k}, \\
        x_{k+1} = \tilde{C}_{k}x_{k} := (E_{k} + C_{k}W_{k})x_{k}.
    \end{align}
\end{subequations}
Given the linear dynamics~\eqref{eq:dynamics_HA_compact}, we postulate a parametric form for the value function in each \texttt{FlipDyn} state as follows:
\begin{equation}
    \label{eq:para_form_al_QC}
    \begin{aligned}
        & V_{k}^{\text{H}}(x,\Xi_{k+1}^{\text{H}}) \Rightarrow V_{k}^{\text{H}}(x) := x^{\tp}P^{\text{H}}_kx, \\
        & V_{k}^{\text{A}}(x, \Xi_{k+1}^{\text{A}}) \Rightarrow V_{k}^{\text{A}}(x) := x^{\tp}P^{\text{A}}_kx,
    \end{aligned}
\end{equation}
where $P^{\text{H}}_{k} \in \mathbb{S}_{+}^{n \times n}$ and $P^{\text{A}}_{k}\in \mathbb{S}_{+}^{n \times n}$ correspond to the \texttt{FlipDyn} states $\alpha = \text{H}$ and $\text{A}$, respectively. 
We impose Assumption~\ref{ast:linear_control_space} to enable factoring out the state $x$  while computing the value function update backward in time. Next, we present the optimal switching policies and value function recursions for each \texttt{FlipDyn} state.

\begin{cor}\label{cor:NE_Val_FDC_H}
    (\textbf{Case $\alpha_k = \text{H}$}) The optimal switching policies of the \texttt{Flip-Team} game~\eqref{eq:obj_def_alpha} for every $k \in \mathcal{K}$, subject to the dynamics~\eqref{eq:dynamics_HA_compact}, with quadratic and takeover costs~\eqref{eq:cost_quad} and \texttt{FlipDyn} dynamics~\eqref{eq:flip_state_H},~\eqref{eq:flip_state_A} are given by:
    \begin{align}\label{eq:TP_quadcost_H}
            \{\pi^{\text{H}*|\text{H}}_{k}, \pi^{\text{A}*|\text{H}}_{k}\} = \begin{cases}
            \{0,0\}, & \text{if } x^{\tp}\left(\tilde{P}_{k+1} + H_{k} + A_{k}\right)x > 0,
            \\[3 pt]
            \{1,1\}, & \text{otherwise}
            \end{cases} 
    \end{align}
    The value is given by:
    \begin{align}\label{eq:Val_quadcost_H}
        P_{k}^{\text{H}} = 
        \begin{cases}
            \begin{aligned}
				& G_k^{\text{H}} + \tilde{B}_{k}^{\tp}P_{k+1}^{\text{H}}\tilde{B}_{k},
            \end{aligned} &\text{if } \tilde{P}_{k+1} + H_{k} + A_{k} \succ \mathbf{0}_{n}, \\[5pt]
            \begin{aligned}
				& G_k^{\text{H}} + \tilde{C}_{k}^{\tp}P_{k+1}^{\text{A}}\tilde{C}_{k} \\ & + H_{k} + A_{k},
            \end{aligned} &\text{otherwise},
		\end{cases} 
    \end{align}
    where $\tilde{P}_{k+1} := \tilde{C}_{k}^{\tp}P_{k+1}^{\text{A}}\tilde{C}_{k} - \tilde{B}_{k}^{\tp}P_{k+1}^{\text{H}}\tilde{B}_{k}$.
    
    \noindent (\textbf{Case $\alpha_k = \text{A}$}) The optimal switching policies are given by:
    \begin{align}
    \begin{split}\label{eq:TP_quadcost_A}
            \{\pi_{k}^{\text{H}*|\text{A}},\pi_{k}^{\text{A}*|\text{A}}\}  = \begin{cases}
            \{0,0\}, & \text{if } \ 
                \begin{matrix}
                    x^{\tp}(\tilde{P}_{k+1})x < \dfrac{x^{\tp}H_{k}x}{p_{k}}
                \end{matrix}, \\[8pt]
            \{1,0\}, & \text{if } \ 
                \begin{matrix}
                    x^{\tp}\tilde{P}_{k+1}x \leq \dfrac{x^{\tp}A_{k}x}{1 - p_{k}}  \\[5pt]
                    x^{\tp}\tilde{P}_{k+1}x \geq \dfrac{x^{\tp}H_{k}x}{p_{k}},
                \end{matrix} \\
            \{1,1\}, & \text{otherwise.}
            \end{cases} 
    \end{split}
    \end{align}
    The value is given by:
    \begin{align}\label{eq:Val_quadcost_A}
        P_{k}^{\text{A}} = 
        \begin{cases}
            \begin{aligned}
				& G_k^{\text{A}} + \tilde{C}_{k}^{\tp}P_{k+1}^{\text{A}}\tilde{C}_{k},
            \end{aligned} &\text{if } \begin{matrix}
                   \tilde{P}_{k+1} \prec \dfrac{H_{k}}{p_{k}}
                \end{matrix}, \\[8pt]
            \begin{aligned}
				& G_k^{\text{A}} + p_{k}\tilde{B}_{k}^{\tp}P^{\text{H}}_{k+1}\tilde{B}_{k} + H_{k}  \\[5pt] & + (1 - p_{k})\tilde{C}_{k}^{\tp}P^{\text{A}}_{k+1}\tilde{C}_{k},
            \end{aligned} &\text{if } \begin{matrix}
                    \tilde{P}_{k+1} \preceq \dfrac{A_{k}}{1 - p_{k}} \\
                    \tilde{P}_{k+1} \succeq \dfrac{H_{k}}{p_{k}},
                \end{matrix} \\
            \begin{aligned}
                & G_k^{\text{A}} + \tilde{B}_{k}^{\tp}P_{k+1}^{\text{H}}\tilde{B}_{k} + \\ & A_{k} + H_{k}
            \end{aligned}
             & \text{otherwise}.
		\end{cases} 
    \end{align}
    The terminal conditions for the recursions~\eqref{eq:Val_quadcost_H} and~\eqref{eq:Val_quadcost_A} are:
    \begin{equation*}
        P_{L+1}^{\text{H}} := G_{L+1}^{\text{H}}, \quad P_{L+1}^{\text{A}} := G_{L+1}^{\text{A}}.
    \end{equation*} 
\end{cor}

\begin{proof}{[Outline]}
    The proof directly follows from Theorem~\ref{th:NE_Val_gen_FDC}. Substituting the parametric form of the value function
    ~\eqref{eq:para_form_al_QC}, dynamics~\eqref{eq:linear_dynamics} and quadratic costs~\eqref{eq:cost_quad} in~\eqref{eq:TP_gen_H} and~\eqref{eq:TP_gen_A}, yields the policies~\eqref{eq:TP_quadcost_H} and~\eqref{eq:TP_quadcost_A}. Similar substitutions yield the value recursions~\eqref{eq:Val_quadcost_H} and~\eqref{eq:Val_quadcost_A}.
\end{proof}

Corollary~\ref{cor:NE_Val_FDC_H} yields offline-computable value function parameters and online-computable switching policies.
We now illustrate these results on scalar and 2D systems, examining sensitivity to the override probability and the effect of state dimension.

\section{Numerical Evaluation}\label{sec:Evaluation}

We illustrate the \texttt{Flip-Team} framework through two examples: a scalar LTI system demonstrating parameter sensitivity, and a 2D vehicle lateral regulation task showing how the framework scales to higher-dimensional systems.

\subsection{Scalar LTI System}\label{subsec:scalar_LTI}

\subsubsection{System Setup}

We evaluate Corollary~\ref{cor:NE_Val_FDC_H} on a discrete-time scalar LTI system over a finite horizon of $L = 40$. The system dynamics under human and autonomous control are characterized by:
\begin{equation*}
    \tilde{B} := E + BK = 1 - 0.9\Delta t, \quad \tilde{C} := E + CW = 1 - 0.6\Delta t,
\end{equation*}
with $E = 1.0$, $B = -0.9\Delta t$, $C = -0.6\Delta t$, and $\Delta t = 0.1$. The closed-loop coefficient $\tilde{B} < \tilde{C}$ indicates superior control performance under human authority, reflecting the human's ability to adapt under challenging conditions, but at a higher operational cost.

The cost parameters are:
\begin{equation*}
    G^{\text{H}} = 1.2, \quad G^{\text{A}} = 1.0, \quad H = 0.35, \quad A = 0.2,
\end{equation*}
where both state and takeover costs are higher for the human agent than for the autonomous agent. We assume time-invariant costs and override probability $p_k = p, \forall k \in \mathcal{K}$.

\subsubsection{Value Function Evolution}

Figures~\ref{fig:val_param_scalar}(a)--(c) show the value function parameters $P^{\text{H}}_k$ and $P^{\text{A}}_k$ for override probabilities $p \in \{0.4, 0.55, 0.7\}$, computed via the backward recursions~\eqref{eq:Val_quadcost_H}--\eqref{eq:Val_quadcost_A}. The corresponding optimal switching policies are shown in Figures~\ref{fig:policy_scalar}(a)--(c).



\textit{Interpretation:} The crossing of $P^{\text{H}}_k$ and $P^{\text{A}}_k$ does not directly trigger switching. Authority transfer occurs only when the threshold conditions in Corollary~\ref{cor:NE_Val_FDC_H} are satisfied, specifically, when $\tilde{P}_{k+1}$ exceeds the scaled takeover costs $H_k/p_k$ or $A_k/(1-p_k)$.

When $p = 0.4$ (low override reliability), the human rarely attempts unilateral override; both agents coordinate to switch early when $\alpha = \text{H}$. When $p = 0.7$ (high override reliability), the human takes control almost immediately under $\alpha = \text{A}$ since successful override is likely. The intermediate case $p = 0.55$ exhibits non-monotonic behavior, with the optimal policy alternating based on the evolving cost-to-go difference.

\begin{figure*}[t]
	\begin{center}
		\subfloat[]{\includegraphics[width = 0.32\linewidth]{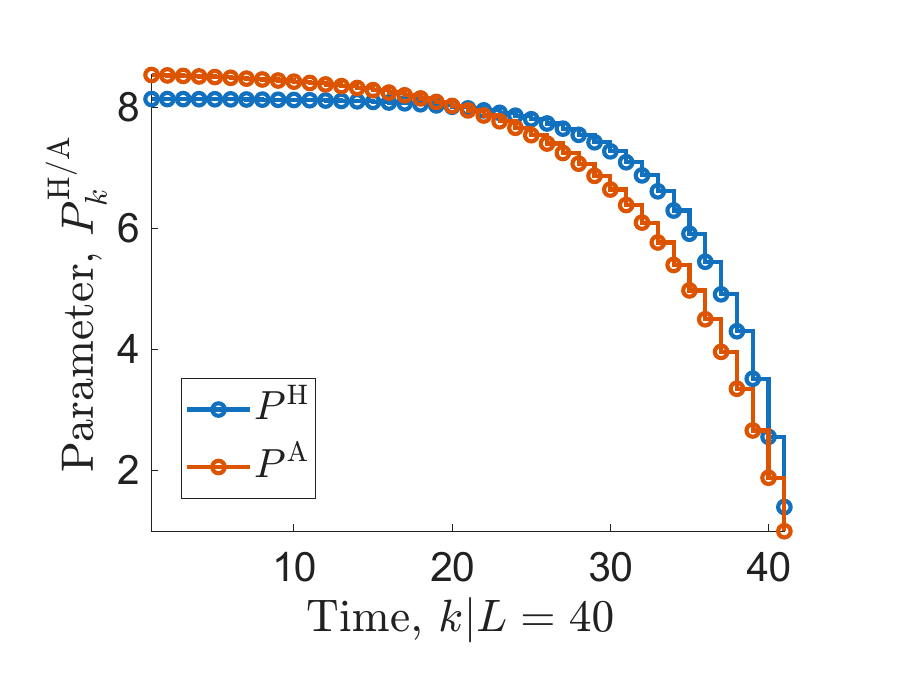}
			\label{fig:SPV_p_04}}
		\subfloat[]{\includegraphics[width = 0.32\linewidth]{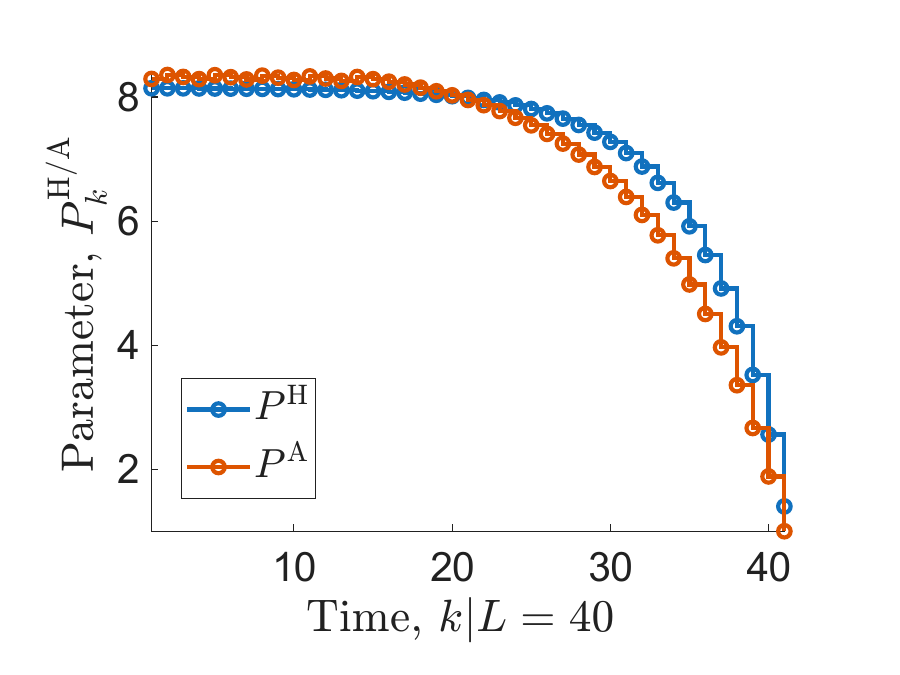}
			\label{fig:SPV_p_55}}
        \subfloat[]{\includegraphics[width = 0.32\linewidth]{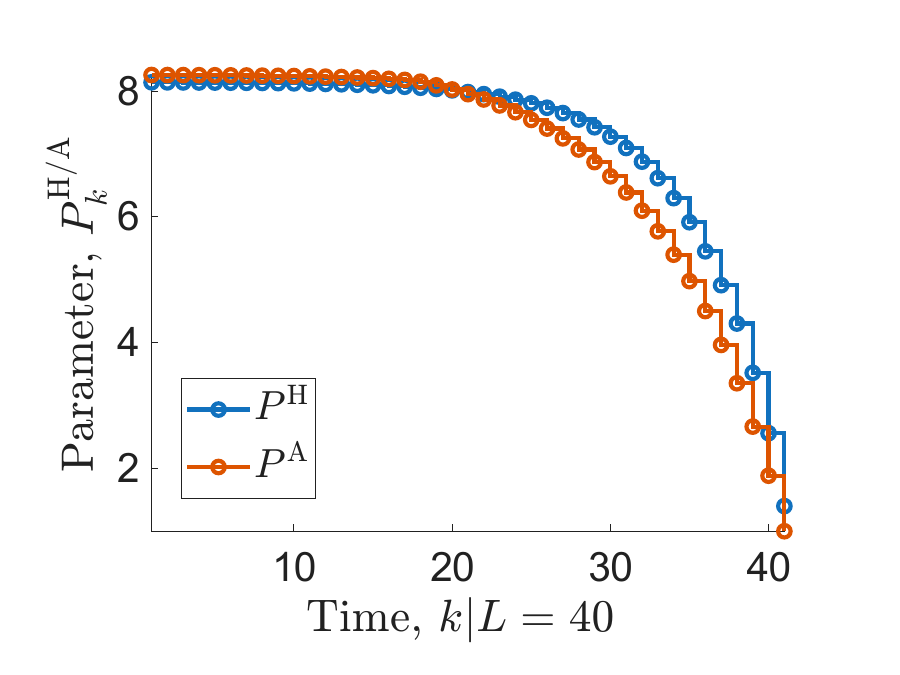}
			\label{fig:SPV_p_7}}
		\caption{\small Value function parameters $P^{\text{H}}_k$ and $P^{\text{A}}_k$ for a scalar LTI system with override probability (a) $p = 0.4$, (b) $p = 0.55$, and (c) $p = 0.7$.}
        \label{fig:val_param_scalar}
	\end{center}
\end{figure*}

\begin{figure*}[h]
	\begin{center}
		\subfloat[]{\includegraphics[width = 0.32\linewidth]{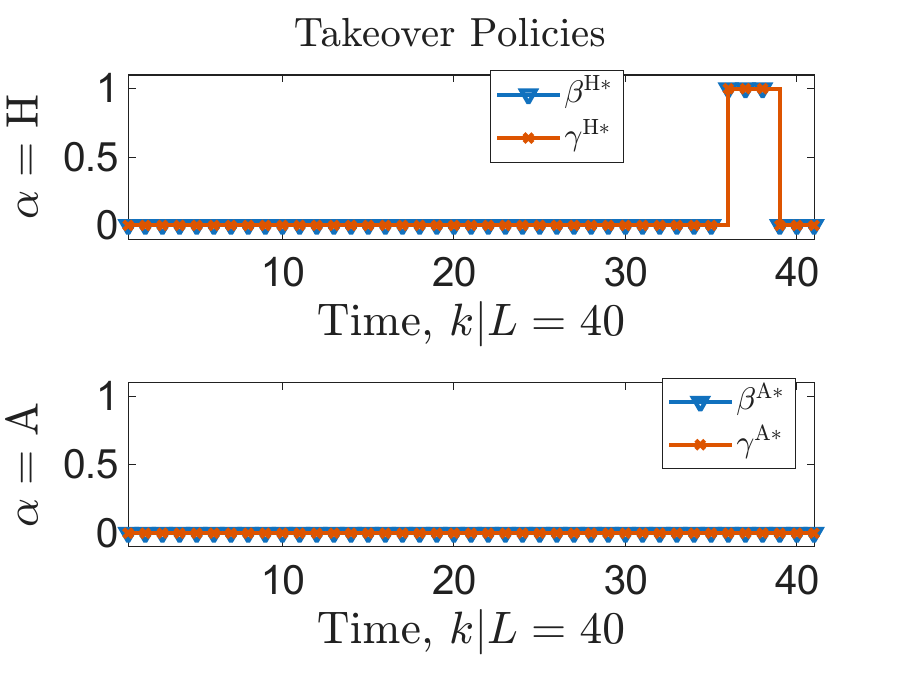}
			\label{fig:TP_p_04}}
		\subfloat[]{\includegraphics[width = 0.32\linewidth]{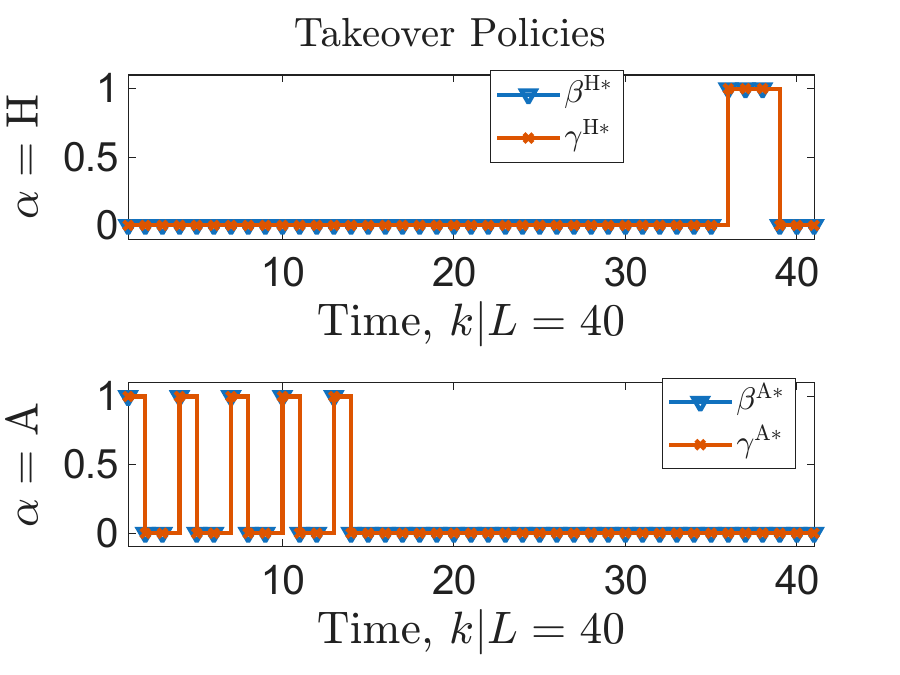}
			\label{fig:TP_p_55}}
        \subfloat[]{\includegraphics[width = 0.32\linewidth]{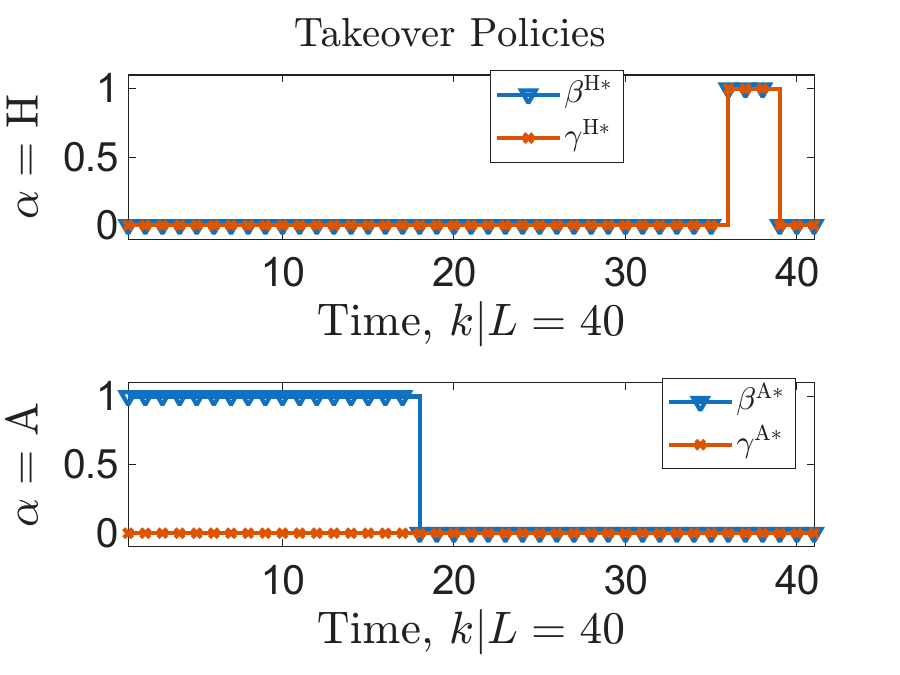}
			\label{fig:TP_p_7}}
		\caption{\small Optimal switching policies $\pi^{\text{H}*}$ and $\pi^{\text{A}*}$ conditioned on authority mode $\alpha \in \{\text{H}, \text{A}\}$ for override probability (a) $p = 0.4$, (b) $p = 0.55$, and (c) $p = 0.7$.}
        \label{fig:policy_scalar}
	\end{center}
\end{figure*}

\subsubsection{Sensitivity to Override Probability}

To address how the switching behavior depends on the override probability, we perform a sensitivity analysis over $p \in [0.1, 0.9]$.

\textit{Critical threshold:} From Corollary~\ref{cor:NE_Val_FDC_H}, the override regime $\{1, 0\}$ requires both (i) a non-empty feasibility band $[H/p, A/(1-p)]$, and (ii) the cost-to-go difference $\tilde{P}_{k+1}$ lying within this band.
Condition (i) yields the necessary threshold:
\begin{equation}\label{eq:critical_threshold}
    p \geq p^* = \frac{H}{H + A} = 0.636.
\end{equation}

Figure~\ref{fig:sensitivity_p} shows the number of time steps at which override is optimal.
For $p < p^*$, the band is empty and no override occurs.
For $p$ slightly above $p^*$, the band exists but is too narrow to capture the realized values of $\tilde{P}_{k+1}$.
Override becomes active around $p \approx 0.68$, where the band has widened sufficiently.
At $p = 0.9$, override is optimal for 32 of 40 time steps.


\begin{figure}[t]
    \centering
    \includegraphics[width=0.8\linewidth]{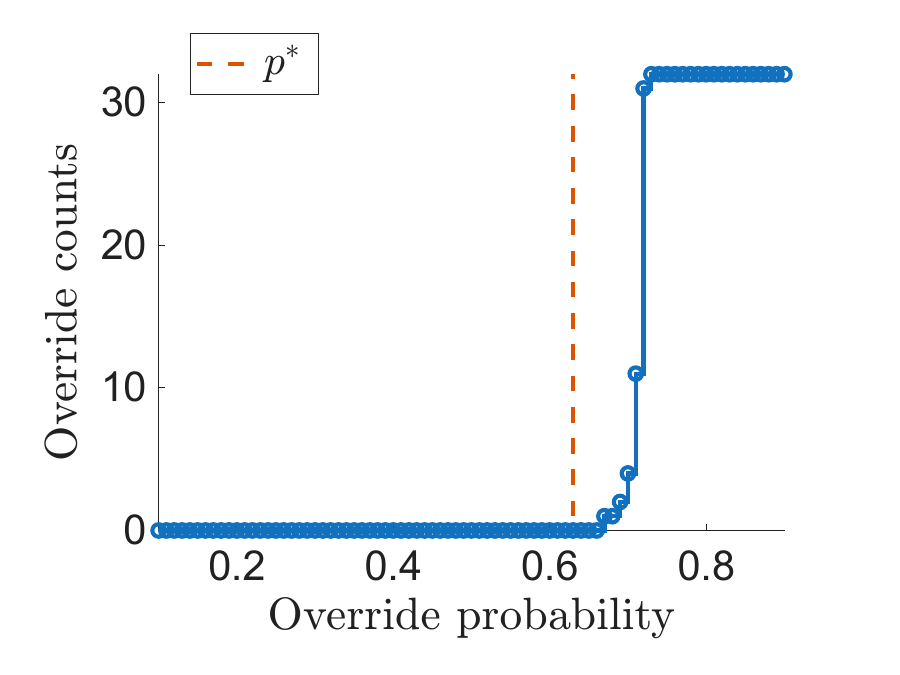}
    \caption{Override count versus override probability $p$. The dashed line marks the necessary threshold $p^* = H/(H+A) = 0.636$. Override occurs only when $p > p^*$ and the cost-to-go difference $\tilde{P}_{k+1}$ lies within the feasible band $[H/p, A/(1-p)]$, which widens as $p$ increases.}
    \label{fig:sensitivity_p}
\end{figure}



\subsection{Vehicle Lateral Regulation}\label{subsec:vehicle_lateral}

We now demonstrate the framework on a 2D vehicle lateral regulation task, illustrating how the switching conditions in Corollary~\ref{cor:NE_Val_FDC_H} operate in higher dimensions.

\subsubsection{System Setup}

Consider a vehicle maintaining lane center on a straight road segment. The state $x_k = [e_y, e_\psi]^\top \in \mathbb{R}^2$ comprises lateral deviation $e_y$ (meters) from lane center and heading error $e_\psi$ (radians). The discrete-time dynamics are:
\begin{equation}\label{eq:vehicle_dynamics}
    x_{k+1} = E x_k + B u_k, \quad E = \begin{bmatrix} 1 & v\Delta t \\ 0 & 1 \end{bmatrix}, \quad B = \begin{bmatrix} 0 \\ \Delta t \end{bmatrix},
\end{equation}
where $v = 15$ m/s is the longitudinal velocity, $\Delta t = 0.1$ s, and $u_k \in \mathbb{R}$ is the steering input. The autonomous agent uses the same dynamics with control matrix $C = B$.

We design linear feedback controllers for each agent:
\begin{itemize}
    \item \textit{Human controller}: $K_{\text{H}} = [-0.8, -1.5]$, reflecting aggressive correction of lateral deviation.
    \item \textit{Autonomous controller}: $K_{\text{A}} = [-0.5, -1.0]$, reflecting smoother, more conservative steering.
\end{itemize}
The resulting closed-loop matrices are:
\begin{equation*}
    \tilde{B} = E + BK_{\text{H}}, \quad \tilde{C} = E + BK_{\text{A}}.
\end{equation*}

The cost matrices are chosen to reflect asymmetric operational characteristics:
\begin{equation*}
    G^{\text{H}} = \begin{bmatrix} 1.5 & 0 \\ 0 & 1.0 \end{bmatrix}, \quad G^{\text{A}} = \begin{bmatrix} 1.0 & 0 \\ 0 & 0.8 \end{bmatrix},
\end{equation*}
with human incurring higher state cost (cognitive effort for precision), and
\begin{equation*}
    H = \begin{bmatrix} 0.4 & 0 \\ 0 & 0.3 \end{bmatrix}, \quad A = \begin{bmatrix} 0.2 & 0 \\ 0 & 0.15 \end{bmatrix},
\end{equation*}
reflecting higher switching cost for human takeover (attention reorientation) than for autonomous handoff.

The horizon is $L = 30$ with override probability $p_k = p = 0.6, \forall k$.

\subsubsection{Results}

Table~\ref{tab:vehicle_comparison} compares the spectral norm of the value function parameter $\lambda_{\max}(P_0)$ across policies.
Since $x^\top P_0 x \leq \lambda_{\max}(P_0)\|x\|^2$ for all initial states $x \neq 0$, this metric provides a uniform cost bound independent of the initial condition.

\begin{table}[t]
    \centering
    \caption{Performance comparison for vehicle lateral regulation ($p = 0.6$, $L = 30$).}
    \label{tab:vehicle_comparison}
    \begin{tabular}{lcc}
        \toprule
        \textbf{Policy} & $\lambda_{\max}(P_0)$ & \textbf{Improvement} \\
        \midrule
        Always-autonomous & 270.19 & (baseline) \\
        Always-human & 222.54 & 17.64\% \\
        Fixed-interval ($N=10$) & 230.20 & 14.8\% \\
        \texttt{Flip-Team} ($\alpha_1 = \text{A}$) & 219.79 & 18.65\% \\
        \texttt{Flip-Team} ($\alpha_1 = \text{H}$) & 219.54 & 18.75\% \\
        \bottomrule
    \end{tabular}
\end{table}

The \texttt{Flip-Team} policy achieves 18.65\% reduction in worst-case cost compared to always-autonomous, demonstrating that optimal authority switching improves performance uniformly across the state space.

\subsubsection{Effect of State Dimension}

In the scalar case, the switching condition $\tilde{P}_{k+1} \geq H/p$ is a scalar comparison.
In higher dimensions, the matrix inequality $\tilde{P}_{k+1} \succeq H/p$ requires positive semidefiniteness of $\tilde{P}_{k+1} - H/p$, introducing state-dependent switching: for some states $x$, the quadratic form $x^\top(\tilde{P}_{k+1} - H/p)x$ may be positive while for others it may be negative.
This means override may be optimal for states aligned with the dominant eigenvectors of $\tilde{P}_{k+1} - H/p$, but not for others, a fundamental structure change compared to scalar case, where switching is state-independent.

\begin{remark}[Computational efficiency]
The value function recursions~\eqref{eq:Val_quadcost_H}--\eqref{eq:Val_quadcost_A} compute $P^{\text{H}}_k, P^{\text{A}}_k \in \mathbb{R}^{n \times n}$ in $O(n^3)$ time per step, independent of the continuous state space. 
\end{remark}
\section{Conclusion}\label{sec:Conclusion}

This paper presented \texttt{Flip-Team}, a cooperative game-theoretic framework for authority switching in shared autonomy.
We formulated the human--autonomy takeover problem as an identical-interest dynamic game with switching dynamics embedded into the system evolution, capturing asymmetric authority where humans retain override capability.

The main contributions are threefold.
First, we formulated the human--autonomy takeover problem as an identical-interest dynamic game in which the switching decision is embedded into the system dynamics, capturing asymmetric authority, stochastic human override, and state-dependent costs within a unified framework.
Second, we established the existence of team-optimal switching policies and characterized them through explicit threshold conditions for when authority transfer occurs in general dynamical systems (Theorem~\ref{th:NE_Val_gen_FDC}).
Third, for linear-quadratic systems, we derived closed-form recursions for the optimal switching policies and value functions, with switching thresholds that are independent of the continuous state (Corollary~\ref{cor:NE_Val_FDC_H}).

Numerical evaluation on scalar and 2D vehicle lateral regulation tasks demonstrated sensitivity of the switching behavior to the override probability $p$, and showed that \texttt{Flip-Team} achieves lower worst-case cost than fixed-authority baselines.

\textit{Limitations:}
The framework assumes identical interests between human and autonomy; extensions to partially misaligned objectives require different solution techniques.
The override probability $p_k$ is treated as known; in deployment, it must be estimated from observed interventions.
The closed-form results are restricted to LQ systems; for nonlinear dynamics, iterative linearization methods (e.g., iLQR) can approximate the optimal policies.

\textit{Future work:}
Directions include online learning of $p_k$ from intervention data, extension to Bayesian games with private information where agents have asymmetric knowledge, and validation on physical platforms with human-in-the-loop experiments.

{\bf Acknowledgements.} This work is supported by the Air Force Office of Scientific Research Grant (AFOSR) Grant AF FA9550-25-1-0274, the National Aeronautics and Space Administration (NASA) under Grant 80NSSC22M0070, and by the National Science Foundation (NSF) under Grants CMMI 2135925, CPS 2311085 and IIS 2331878.

\bibliographystyle{IEEEtran}
\bibliography{IEEEabrv, myRefs}

@article{abbinkTopologySharedControl2018,
  title = {{A {{Topology}} of {{Shared Control Systems}}---{{Finding Common Ground}} in {{Diversity}}}},
  author = {Abbink, David A. and Carlson, Tom and Mulder, Mark and {de Winter}, Joost C. F. and Aminravan, Farzad and Gibo, Tricia L. and Boer, Erwin R.},
  year = {2018},
  month = oct,
  journal = {IEEE Transactions on Human-Machine Systems},
  volume = {48},
  number = {5},
  pages = {509--525},
  issn = {2168-2305},
  url = {https://doi.org/10.1109/THMS.2018.2791570},
  urldate = {2025-09-30}
}

@article{loseyReviewIntentDetection2018,
  title = {{A {{Review}} of {{Intent Detection}}, {{Arbitration}}, and {{Communication Aspects}} of {{Shared Control}} for {{Physical Human}}--{{Robot Interaction}}}},
  author = {Losey, Dylan P. and McDonald, Craig G. and Battaglia, Edoardo and O'Malley, Marcia K.},
  year = {2018},
  month = feb,
  journal = {Applied Mechanics Reviews},
  volume = {70},
  number = {010804},
  issn = {0003-6900},
  url = {https://doi.org/10.1115/1.4039145},
  urldate = {2025-09-30}
}

@article{aarnoMotionIntentionRecognition2008,
  title = {{Motion Intention Recognition in Robot Assisted Applications}},
  author = {Aarno, Daniel and Kragic, Danica},
  year = {2008},
  month = aug,
  journal = {Robotics and Autonomous Systems},
  volume = {56},
  number = {8},
  pages = {692--705},
  issn = {0921-8890},
  url = {10.1016/j.robot.2007.11.005},
  urldate = {2025-09-30}
}

@misc{alambeigiCrashThemesAutomated2020,
  title = {{Crash {{Themes}} in {{Automated Vehicles}}: {{A Topic Modeling Analysis}} of the {{California Department}} of {{Motor Vehicles Automated Vehicle Crash Database}}}},
  shorttitle = {Crash {{Themes}} in {{Automated Vehicles}}},
  author = {Alambeigi, Hananeh and McDonald, Anthony D. and Tankasala, Srinivas R.},
  year = {2020},
  month = jan,
  number = {arXiv:2001.11087},
  eprint = {2001.11087},
  primaryclass = {stat},
  publisher = {arXiv},
  url = {https://arxiv.org/abs/2001.11087},
  urldate = {2025-09-30},
  archiveprefix = {arXiv}
}

@article{gouraudAutopilotMindWandering2017,
  title = {{Autopilot, {{Mind Wandering}}, and the {{Out}} of the {{Loop Performance Problem}}}},
  author = {Gouraud, Jonas and Delorme, Arnaud and Berberian, Bruno},
  year = {2017},
  month = oct,
  journal = {Frontiers in Neuroscience},
  volume = {11},
  publisher = {Frontiers},
  issn = {1662-453X},
  url = {https://doi.org/10.3389/fnins.2017.00541},
  urldate = {2025-09-30},
  langid = {english}
}

@inproceedings{banikFlipDynGraphsResource2025,
  title = {{{{FlipDyn}} in~{{Graphs}}: {{Resource Takeover Games}} in~{{Graphs}}}},
  shorttitle = {{{FlipDyn}} in~{{Graphs}}},
  booktitle = {Decision and {{Game Theory}} for {{Security}}},
  author = {Banik, Sandeep and Bopardikar, Shaunak D. and Hovakimyan, Naira},
  editor = {Sinha, Arunesh and Fu, Jie and Zhu, Quanyan and Zhang, Tao},
  year = {2025},
  pages = {220--239},
  publisher = {Springer Nature Switzerland},
  address = {Cham},
  url = {https://doi.org/10.1007/978-3-031-74835-6\_11},
  isbn = {978-3-031-74835-6},
  langid = {english}
}

@article{vandijkFlipItGameStealthy2013,
  title = {{{{FlipIt}}: {{The Game}} of ``{{Stealthy Takeover}}''}},
  shorttitle = {{{FlipIt}}},
  author = {{van Dijk}, Marten and Juels, Ari and Oprea, Alina and Rivest, Ronald L.},
  year = {2013},
  month = oct,
  journal = {Journal of Cryptology},
  volume = {26},
  number = {4},
  pages = {655--713},
  issn = {1432-1378},
  url = {https://doi.org/10.1007/s00145-012-9134-5},
  urldate = {2025-08-19},
  langid = {english}
}

@article{loseyPhysicalInteractionCommunication2022,
  title = {{Physical Interaction as Communication: {{Learning}} Robot Objectives Online from Human Corrections}},
  shorttitle = {Physical Interaction as Communication},
  author = {Losey, Dylan P. and Bajcsy, Andrea and O’Malley, Marcia K. and Dragan, Anca D.},
  date = {2022-01-01},
  journaltitle = {The International Journal of Robotics Research},
  volume = {41},
  number = {1},
  pages = {20--44},
  publisher = {SAGE Publications Ltd STM},
  issn = {0278-3649},
  doi = {10.1177/02783649211050958},
  url = {https://doi.org/10.1177/02783649211050958},
  urldate = {2025-08-19},
  langid = {english}
}

@article{scerriAdjustableAutonomyReal2002,
  title = {{Towards {{Adjustable Autonomy}} for the {{Real World}}}},
  author = {Scerri, P. and Pynadath, D. V. and Tambe, M.},
  year = {2002},
  month = sep,
  journal = {Journal of Artificial Intelligence Research},
  volume = {17},
  pages = {171--228},
  issn = {1076-9757},
  url = {https://doi.org/10.1613/jair.1037},
  urldate = {2025-08-19},
  copyright = {Copyright (c)},
  langid = {english}
}

@article{draganPolicyblendingFormalismShared2013,
  title = {{A Policy-Blending Formalism for Shared Control}},
  author = {Dragan, Anca D and Srinivasa, Siddhartha S},
  year = {2013},
  month = jun,
  journal = {The International Journal of Robotics Research},
  volume = {32},
  number = {7},
  pages = {790--805},
  publisher = {SAGE Publications Ltd STM},
  issn = {0278-3649},
  url = {https://doi.org/10.1177/0278364913490324},
  urldate = {2025-04-21}
}

@inproceedings{nikolaidisGameTheoreticModelingHuman2017,
  title = {{Game-{{Theoretic Modeling}} of {{Human Adaptation}} in {{Human-Robot Collaboration}}}},
  booktitle = {Proceedings of the 2017 {{ACM}}/{{IEEE International Conference}} on {{Human-Robot Interaction}}},
  author = {Nikolaidis, Stefanos and Nath, Swaprava and Procaccia, Ariel D. and Srinivasa, Siddhartha},
  year = {2017},
  month = mar,
  series = {{{HRI}} '17},
  pages = {323--331},
  publisher = {Association for Computing Machinery},
  address = {New York, NY, USA},
  doi = {https://doi.org/10.1145/2909824.3020253},
  urldate = {2025-08-19},
  isbn = {978-1-4503-4336-7}
}

@article{liDifferentialGameTheory2019,
  title = {{Differential Game Theory for Versatile Physical Human--Robot Interaction}},
  author = {Li, Y. and Carboni, G. and Gonzalez, F. and Campolo, D. and Burdet, E.},
  year = {2019},
  month = jan,
  journal = {Nature Machine Intelligence},
  volume = {1},
  number = {1},
  pages = {36--43},
  publisher = {Nature Publishing Group},
  issn = {2522-5839},
  url = {https://doi.org/10.1038/s42256-018-0010-3},
  urldate = {2025-08-19},
  copyright = {2019 The Author(s), under exclusive licence to Springer Nature Limited},
  langid = {english}
}

@misc{jeonSharedAutonomyLearned2020,
  title = {{Shared {{Autonomy}} with {{Learned Latent Actions}}}},
  author = {Jeon, Hong Jun and Losey, Dylan P. and Sadigh, Dorsa},
  year = {2020},
  month = may,
  number = {arXiv:2005.03210},
  eprint = {2005.03210},
  primaryclass = {cs},
  publisher = {arXiv},
  url = {https://doi.org/10.48550/arXiv.2005.03210},
  urldate = {2025-08-19},
  archiveprefix = {arXiv}
}

@misc{reddySharedAutonomyDeep2018,
  title = {{Shared {{Autonomy}} via {{Deep Reinforcement Learning}}}},
  author = {Reddy, Siddharth and Dragan, Anca D. and Levine, Sergey},
  year = {2018},
  month = may,
  number = {arXiv:1802.01744},
  eprint = {1802.01744},
  primaryclass = {cs},
  publisher = {arXiv},
  url = {https://doi.org/10.48550/arXiv.1802.01744},
  urldate = {2025-08-19},
  archiveprefix = {arXiv}
}

@inproceedings{banikFlipDynGameResource2022a,
  title = {{{FlipDyn}: {A} Game of Resource Takeovers in Dynamical Systems}},
  shorttitle = {{{FlipDyn}}},
  booktitle = {2022 {{IEEE}} 61st {{Conference}} on {{Decision}} and {{Control}} ({{CDC}})},
  author = {Banik, Sandeep and Bopardikar, Shaunak D.},
  year = {2022},
  month = dec,
  pages = {2506--2511},
  issn = {2576-2370},
  url = {https://doi.org/10.1109/CDC51059.2022.9992387},
  urldate = {2025-08-19}
}

@inproceedings{sadighPlanningAutonomousCars2016,
  title = {{Planning for {{Autonomous Cars}} That {{Leverage Effects}} on {{Human Actions}}}},
  booktitle = {Robotics: {{Science}} and {{Systems XII}}},
  author = {Sadigh, Dorsa and Sastry, Shankar and A. Seshia, Sanjit and D. Dragan, Anca},
  year = {2016},
  publisher = {{Robotics: Science and Systems Foundation}},
  url = {https://doi.org/10.15607/RSS.2016.XII.029},
  urldate = {2025-08-18},
  isbn = {978-0-9923747-2-3},
  langid = {english}
}

@inproceedings{nikolaidisHumanRobotMutualAdaptation2017,
  title = {{Human-{{Robot Mutual Adaptation}} in {{Shared Autonomy}}}},
  booktitle = {Proceedings of the 2017 {{ACM}}/{{IEEE International Conference}} on {{Human-Robot Interaction}}},
  author = {Nikolaidis, Stefanos and Zhu, Yu Xiang and Hsu, David and Srinivasa, Siddhartha},
  year = {2017},
  month = mar,
  pages = {294--302},
  publisher = {ACM},
  address = {Vienna Austria},
  url = {https://doi.org/10.1145/2909824.3020252},
  urldate = {2025-08-18},
  isbn = {978-1-4503-4336-7},
  langid = {english}
}

@article{javdaniSharedAutonomyHindsight2015,
author = {Shervin Javdani and Henny Admoni and Stefania Pellegrinelli and Siddhartha S. Srinivasa and J. Andrew Bagnell},
title ={{{Shared Autonomy via Hindsight Optimization for Teleoperation and Teaming}}},

journal = {The International Journal of Robotics Research},
volume = {37},
number = {7},
pages = {717-742},
year = {2018},
doi = {10.1177/0278364918776060},
URL = {https://doi.org/10.1177/0278364918776060},
eprint = {https://doi.org/10.1177/0278364918776060}
}

@article{zhu2015game,
  title={{{Game-Theoretic Methods for Robustness, Security, and Resilience of Cyberphysical Control Systems: Games-in-Games Principle for Optimal Cross-Layer Resilient Control Systems}}},
  author={Zhu, Quanyan and Basar, Tamer},
  journal={IEEE {C}ontrol Systems Magazine},
  volume={35},
  number={1},
  pages={46--65},
  year={2015},
  publisher={IEEE},
  url = {https://doi.org/10.1109/MCS.2014.2364710},
}

@book{hespanha2017noncooperative,
	title={{Noncooperative game theory: An introduction for engineers and computer scientists}},
	author={Hespanha, Jo{\~a}o P},
	year={2017},
	publisher={Princeton University Press}
}

\end{document}